\documentclass[letterpaper]{article} %

\pdfoutput=1

\usepackage{aaai24}  %
\usepackage{times}  %
\usepackage{helvet}  %
\usepackage{courier}  %
\usepackage[hyphens]{url}  %
\usepackage{graphicx} %
\usepackage{natbib}  %
\usepackage{caption} %
\usepackage{algorithm}
\usepackage{algorithmic}
\usepackage{subcaption}
\usepackage{tikz}
\usetikzlibrary{shapes,snakes}

\usepackage{booktabs}

\usepackage{newfloat}
\usepackage{listings}
\DeclareCaptionStyle{ruled}{labelfont=normalfont,labelsep=colon,strut=off} %
\floatstyle{ruled}
\newfloat{listing}{tb}{lst}{}
\floatname{listing}{Listing}
\usepackage{amsmath}
\usepackage{amssymb}
\usepackage{amsthm}
\usepackage{cleveref}
\usepackage{thm-restate}

\newtheorem{definition}{Definition}
\newtheorem{example}{Example}

\newtheorem{remark}{Remark}

\newcommand{\primarykeywords}{
None
}

\usepackage[switch, modulo]{lineno}

\title{Non-Deterministic Planning for Hyperproperty Verification}\author{
    Raven Beutner,
    Bernd Finkbeiner
}
\affiliations{
    CISPA Helmholtz Center for Information Security, Germany
}

\usepackage{amsmath}
\usepackage{amsfonts}
\usepackage{dsfont}

\newcommand{\ap}{\mathit{AP}}
\newcommand{\pathVars}{\mathcal{V}}

\newcommand{\directions}{\mathbb{D}}

\newcommand{\ldot}{\mathpunct{.}}

\newcommand{\quant}{\mathds{Q}}

\newcommand{\nat}{\mathbb{N}}

\newcommand{\calT}{\mathcal{T}}
\newcommand{\calG}{\mathcal{G}}
\newcommand{\calA}{\mathcal{A}}

\DeclareMathOperator{\ltlN}{\normalfont\textsf{X}}
\DeclareMathOperator{\ltlG}{\normalfont\textsf{G}}
\DeclareMathOperator{\ltlF}{\normalfont\textsf{F}}

\DeclareMathOperator{\ltlU}{\normalfont\textsf{U}}

\newcommand{\agents}{I}
\newcommand{\paths}{\mathit{Paths}}

\newcommand{\interactions}{\mathit{Exec}}

\newcommand{\tool}{\texttt{HyPlan}}

\newif\iffullversion
\fullversiontrue

\newcommand{\ifFull}[2]{\iffullversion#1\else#2\fi}

\begin{document}

\maketitle

\begin{abstract}
Non-deterministic planning aims to find a policy that achieves a given objective in an environment where actions have uncertain effects, and the agent -- potentially -- only observes parts of the current state.
Hyperproperties are properties that relate multiple paths of a system and can, e.g., capture security and information-flow policies.
Popular logics for expressing temporal hyperproperties -- such as HyperLTL -- extend LTL by offering selective quantification over executions of a system.
In this paper, we show that planning offers a powerful intermediate language for the automated verification of hyperproperties. 
Concretely, we present an algorithm that, given a HyperLTL verification problem, constructs a non-deterministic multi-agent planning instance (in the form of a QDec-POMDP) that, when admitting a plan, implies the satisfaction of the verification problem.
We show that for large fragments of HyperLTL, the resulting planning instance corresponds to a classical, FOND, or POND planning problem.
We implement our encoding in a prototype verification tool and report on encouraging experimental results.

\end{abstract}

\section{Introduction}\label{sec:intro}

AI planning is the task of finding a policy (aka.~plan) that ensures that a specified goal is reached. 
In this paper, we present an exciting new application of planning: the automated verification of \emph{hyperproperties}. 

\paragraph{Hyperproperties and HyperLTL.}

Hyperproperties generalize traditional trace properties by relating \emph{multiple} executions of a system \cite{ClarksonS08}.
A trace property -- specified, e.g.,  in LTL -- reasons about \emph{individual} executions in isolation, which falls short in many applications.
For example, assume we model the dynamics of a system as a transition system over atomic propositions $\{o, h, l\}$, and want to specify that the output $o$ of the system only depends on the low-security input $l$ and does not leak information about the secret input $h$. 
We cannot specify this as a trace property in LTL; we need to relate multiple executions to observe how different inputs impact the output. 
HyperLTL extends LTL with explicit quantification over executions \cite{ClarksonFKMRS14}, and allows for the specification of such a property.
For example, we can express observational determinism (OD) \cite{ZdancewicM03} as the following HyperLTL formula:
\begin{align}\label{eq:ni}
	\forall \pi\ldot \forall \pi'\ldot (l_{\pi} \leftrightarrow l_{\pi'}) \to \ltlG (o_{\pi} \leftrightarrow o_{\pi'})\tag{OD}
\end{align}
This formula states that on any \emph{pair} of executions $\pi, \pi'$ with identical low-security input, the output is (globally) the same. 
In other words, the output of the system behaves deterministically w.r.t.~the low-security input.
For non-deterministic systems, \ref{eq:ni} is often too strict, as any given low-security input might lead to multiple outputs. 
A relaxed notation -- called non-inference (NI) \cite{McLean94} -- can be expressed in HyperLTL as follows:
\begin{align}\label{eq:gni}
	\forall \pi\ldot \exists \pi'\ldot \ltlG \big( (o_{\pi} \leftrightarrow o_{\pi'}) \land (l_{\pi} \leftrightarrow l_{\pi'}) \land  \neg h_{\pi'}\big) \tag{NI}
\end{align}
That is, for any execution $\pi$, there \emph{exists} another execution $\pi'$ that has the same low-security behavior (via propositions $o$ and $l$), despite having a fixed ``dummy'' high-security input (in our case, we require that $h$ is always false, i.e., never holds on $\pi'$).
If \ref{eq:gni} holds, a low-security attacker cannot distinguish any high-security input from the dummy input.

\paragraph{HyperLTL Verification as Planning.}

Our goal is to automatically verify that a finite-state system $\calT$ satisfies a HyperLTL formula $\varphi$.
We introduce a novel verification approach that leverages the advanced methods developed within the planning community. 
Concretely, we present a reduction that soundly converts a HyperLTL verification problem into a planning problem.
Depending on the HyperLTL formula, our encoding uses several advanced features supported by modern planning frameworks, such as uncertain action effects (non-determinism) \cite{CimattiPRT03}, partial observations \cite{BertoliCRT06}, and multiple agents.
We show that -- by carefully combining these features -- we obtain a planning problem that is sound w.r.t.~the HyperLTL semantics: every plan can be translated back into a validity witness for the original verification problem.
As a consequence, our encoding allows us to utilize mature planning tools for the automated verification of HyperLTL properties.
We implement our encoding as a prototype and report on encouraging results using off-the-shelf planners.

\section{High-Level Overview}\label{sec:overview}

Before proceeding with a formal construction, we provide a high-level intuition of our encoding.
In HyperLTL, we can quantify over the executions of a system (as seen informally in \ref{eq:ni} and \ref{eq:gni}).
The overarching idea in our encoding is to let the planning agent(s) control all \emph{existentially quantified} executions, such that any valid plan directly corresponds to a witness for the existentially quantified executions.

\paragraph{Verification as Planning.}
As an example, assume we want to verify that \ref{eq:ni} does \emph{not} hold on a given system $\calT$, i.e., we want to find concrete executions $\pi, \pi'$ that violate the body of \ref{eq:ni}.
We can interpret this as a classical (single-agent) planning problem: each planning state maintains two locations in $\calT$, one for $\pi$ and one for $\pi'$, and, in each step, the actions update the locations for $\pi, \pi'$ by moving along the transitions in $\calT$.
The planning objective is to construct executions for $\pi, \pi'$ that \emph{violate} $(l_{\pi} \leftrightarrow l_{\pi'}) \to \ltlG (o_{\pi} \leftrightarrow o_{\pi'})$.
Any successful plan (i.e., sequence of transitions) then directly corresponds to concrete paths $\pi, \pi'$ disproving \ref{eq:ni}.

\paragraph{Non-Deterministic Planning.}

Verification becomes more interesting when the HyperLTL formula contains quantifier alternations, such as \ref{eq:gni}. 
Following the above intuition, a plan should provide a concrete witness for (the existentially quantified) $\pi'$, but -- this time -- we need to consider \emph{all possible} executions for (the universally quantified) $\pi$. 
Our idea is that we can approximate this behavior by viewing it as a fully observable non-deterministic (FOND) planning problem; intuitively, a plan controls the behavior of $\pi'$ while the behavior of $\pi$ is non-deterministic. 
That is, each action determines a successor location for $\pi'$ but also non-deterministically updates the location of $\pi$.
The agent's objective is to ensure that the resulting paths $\pi, \pi'$, together, satisfy $ \ltlG \big((o_{\pi} \leftrightarrow o_{\pi'}) \land (l_{\pi} \leftrightarrow l_{\pi'}) \land \neg h_{\pi'}\big)$ (the LTL body of \ref{eq:gni}).
Any plan (which is now \emph{conditional} on the non-deterministic outcomes) thus defines a concrete execution for $\pi'$, depending on the concrete execution for $\pi$.

\paragraph{Planning Under Partial Observations.}

In \ref{eq:gni}, $\pi'$ is quantified \emph{after} $\pi$, so the action sequence that defines the behavior of $\pi'$ can be based on the behavior of $\pi$.
This changes when quantifiers \emph{trail} existential quantification, e.g., in a formula of the form $\exists \pi \ldot\forall \pi'$. 
Here, we follow the same idea as before but ensure that the actions controlling $\pi$ are \emph{independent} of the current location of $\pi'$, i.e., the agent(s) act in a \emph{partially observable} non-deterministic (POND) domain.

\section{Related Work}\label{sec:relatedWork}

\paragraph{Non-Deterministic Planning.}

Non-deterministic planning provides a powerful intermediate language that encompasses problems such as reactive synthesis \cite{CamachoBMM18}, controller synthesis in MDPs, epistemic planning \cite{EngesserM20}, and generalized planning \cite{HuG11}.
Consequently, many methods and tools have been developed \cite{PereiraPMG22,MessaP23,CamachoTMBM17,GeffnerG18,RodriguezBSG21,MuiseMB12,KuterNRG08}, with some also supporting planning in partially observable domains \cite{BertoliCRT06,CimattiPRT03,BonetG11}.

\paragraph{HyperLTL Verification.}

Model checking of HyperLTL on finite-state transition systems is decidable, and existing complete algorithms utilize expensive automata complementations or inclusion checks \cite{FinkbeinerRS15,BeutnerF23}. 
There also exist cheaper (but incomplete) methods based, e.g., on a bounded model-checking \cite{HsuSB21}. 
For $\forall^*\exists^*$ HyperLTL properties (i.e., properties where no universal quantified appears below an existential quantifier), our encoding is closely related to game-based (or, equivalently, simulation-based) approaches \cite{BeutnerF22,BeutnerF22b,HsuSSB23}, which interpret verification as a game between universal and existential quantifiers.  
In fact, non-deterministic planning can be seen as a specialized form of turn-based games \cite{KissmannE09}.
Crucially, the size of an (explicit-state) game-based approach scales \emph{exponentially} in the number of quantified executions \cite{BeutnerF22}, making it impractical for larger instances. 
In contrast, the planning-based approach in this paper can describe the problem in a factored representation and let the planner determine how to best explore the state space.
Moreover, -- by utilizing partial observations -- our planning-based encoding is applicable to HyperLTL formulas with arbitrary quantifier prefixes, not only $\forall^*\exists^*$ formulas.

\section{Preliminaries}

For a set $X$, we write $X^+$ for the set of finite non-empty sequences over $X$, and $X^\omega$ for the set of infinite sequences. 

\subsection{Non-Deterministic Planning}

As a basic planning model, we use Qualitative Dec-POMDPs (QDec-POMDP), a general framework that encompasses multiple agents, non-deterministic effects, and partial observations \cite{BrafmanSZ13}.

\begin{definition}
	A QDec-POMDP is a tuple $\calG = (I, S, S_0, \allowbreak \{A_i\}, \delta, \{\Omega_i\}, \{\omega_i\}, G)$, where 
	$I$ is a finite set of agents; $S$ is a finite set of states; $S_0 \subseteq S$ is a set of initial states; 
	for each $i \in I$, $A_i$ is a finite set of actions, and we define $\vec{A} := \otimes_{i \in I} A_i$ as the set of joint actions;
	$\delta : S \times \vec{A} \to 2^S$ is a (non-deterministic) transition function;
	for each $i \in I$, $\Omega_i$ is a finite set of observations; $\omega_i : S \to \Omega_i$ defines $i$'s local observation; and $G \subseteq S$ is a set of goal states.
\end{definition}

We write $\{a_i\} \in \vec{A}$ for the joint action where each agent $i \in I$ chooses action $a_i$.
A \emph{local policy} for an agent $i \in \agents$, is a conditional plan that picks an action based on the history of observations, i.e., a function $f_i : \Omega_i^+ \to A_i$.
We can represent a policy $f_i$ as an (infinite) tree of degree $|\Omega_i|$ where nodes are labeled with elements from $A_i$.
A \emph{joint policy} $\{f_i\}$ assigns each agent $i \in I$ a local policy $f_i$.
A finite path $p \in S^+$ is \emph{compatible with $\{f_i\}$} if and only if \textbf{(1)} $p(0) \in S_0$ (i.e., the path starts in an initial state), and \textbf{(2)} for every $0 \leq j < |p| - 1$, $p(j+1) \in \delta(p(j), \{a_i\})$, where $a_i := f_i (\omega_i(p(0)) \cdots \omega_i(p(j)))$ for every $i \in \agents$.
That is, in the $j$th step, we compute the joint action $\{a_i\}$, where each $a_i$ is determined by policy $f_i$ based on the past observations made by $i$ on the prefix $p(0) \cdots p(j)$.
We write $\interactions(\{f_i\}) \subseteq S^+$ for the set of all $\{f_i\}$-compatible paths.

The objective of the agents is to reach a goal state in $G$. 
Following \citet{CimattiPRT03}, we distinguish between different levels of reachability.
A policy is a \emph{weak plan} if \emph{some} $p \in \interactions(\{f_i\})$ reaches a state in $G$, i.e., $\{f_i\}$ can reach the goal provided the non-determinism is resolved favorably. 
A policy is a \emph{strong plan} if there exists an $N \in \nat$ such that \emph{every} $p \in \interactions(\{f_i\})$ with $|p| \geq N$ reaches $G$, i.e., the goal is guaranteed to be reached, irrespective of non-deterministic outcomes.
Finally, a policy is a \emph{strong cyclic plan} if, for every $p \in \interactions(\{f_i\})$, either $p$ reaches a state in $G$ or there exists some $p' \in \interactions(\{f_i\})$ that extends $p$ (i.e., $p$ is a \emph{prefix} of $p'$) and $p'$ reaches a state in $G$.\footnote{A strong cyclic plan is one that always \emph{preserves the possibility of reaching a goal state}, i.e., at every point during the plan's execution, the non-determinism can be resolved favorably such that the goal is reached. Our definition expresses exactly this: either $p$ already reaches $G$ or some extension of $p$ \emph{can} reach the goal. This definition is equivalent to the one of \citet{CimattiPRT03}.
}

\subsection{Transition Systems}

\begin{figure*}[!t]
	\begin{subfigure}{0.1\linewidth}
		\centering
		\scalebox{0.9}{
		\begin{tikzpicture}
			\node[circle, inner sep=2pt,draw, thick, label=right:{\small $\{a\}$}] at (0,0) (n0) {\small $l_A$};
			
			\node[circle, inner sep=2pt, draw, thick, label=right:{\small $\emptyset$}] at (0,-2) (n1) {\small $l_B$};
			
			\draw[->, thick, bend left=15] (n0) to node[right,align=center] {\small $d_B$} (n1);
			\draw[->, thick, bend left=15] (n1) to node[left,align=center] {\small $d_A$} (n0);
			
			\draw [->,thick] (n0) edge[loop above] node[above,align=center] {\small $d_A$}(n0);
			\draw [->,thick] (n1) edge[loop below] node[below,align=center] {\small $d_B$}(n1);
			
			\draw[->, thick] (-0.4, 0.4) to[] (n0);	
			\end{tikzpicture}
		}
		
		\subcaption{}\label{fig:ts}
	\end{subfigure}%
	\begin{subfigure}{0.29\linewidth}
		\centering
		
		\scalebox{0.9}{
		\begin{tikzpicture}
			\node[regular polygon,regular polygon sides=7, inner sep=1pt,draw, thick] at (0,0) (n0) {\small $q_0$};
			
			\node[regular polygon,regular polygon sides=7, inner sep=1pt,draw, thick] at (1.5,-1.5) (n1) {\small $q_1$};
			
			\node[regular polygon,regular polygon sides=7, inner sep=1pt,draw, thick] at (1.5,1.5) (n2) {\small $q_2$};
			
			\node[regular polygon,regular polygon sides=7, inner sep=1pt,draw, thick,double] at (3,0) (n3) {\small $q_3$};
			
			\draw[->, thick] (n0) to[out=270,in=180] node[below,align=center,sloped] {\small $a_{\pi_1}$} (n1);
			\draw[->, thick] (n0) to[out=90,in=180] node[above,align=center,sloped] {\small $\neg a_{\pi_1}$} (n2);

			\draw[->, thick,bend left=20] (n1) to node[above,align=center,sloped,rotate=180] {\small $a_{\pi_2}  \land  \neg a_{\pi_1}$} (n2);
			\draw[->, thick,bend left=20] (n2) to node[above,align=center,sloped,rotate=180] {\small $\neg a_{\pi_2}  \land a_{\pi_1}$} (n1);
			
			\draw [->,thick] (n1) edge[loop below] node[below,align=center] {\small $a_{\pi_2} \land  a_{\pi_1}$}(n1);
			\draw [->,thick] (n2) edge[loop above] node[above,align=center,yshift=-1mm] {\small $\neg a_{\pi_2} \land \neg a_{\pi_1}$}(n2);
			
			\draw[->, thick] (n1) to[out=0,in=270] node[below,align=center,sloped] {\small $\neg a_{\pi_2}$} (n3);
			\draw[->, thick] (n2) to[out=0,in=90] node[above,align=center,sloped] {\small $a_{\pi_2}$} (n3);
			
			\draw [->,thick] (n3) edge[loop right] node[right,align=center] {\small $\top$}(n3);
			
			\draw[->, thick] (-0.4, 0.4) to[] (n0);	
		\end{tikzpicture}
	}
		
		\subcaption{}\label{fig:dsa}
	\end{subfigure}%
	\begin{subfigure}{0.61\linewidth}
		\centering
		\scalebox{0.89}{
		\begin{tikzpicture}
			\node[rectangle, inner sep=3pt,draw, thick,rounded corners=2pt] at (0,0) (n0) {\small $\langle l_A, l_A, q_0\rangle$};
			
			\draw[->, thick] (-0.5, 0.5) to[] (n0);	
			
			\node[circle,fill, inner sep=2pt,label={[xshift=1mm]left:{\small$d_A$}}] at (-3, -0.5) (c0) {};

			\node[rectangle, inner sep=3pt,draw, thick,rounded corners=2pt] at (-3, 1) (n1) {\small $\langle l_A, l_A, q_1\rangle$};
			\node[rectangle, inner sep=3pt,draw, thick,rounded corners=2pt] at (-3,-3) (n2) {\small $\langle l_B, l_A, q_1\rangle$};
			
			\draw[->, thick,bend left=20] (c0) to (n1);
			\draw[->, thick] (c0) to (n2);
			
			\draw[->, thick,bend left=20] (n1) to node[sloped] {\color{orange}$\blacktriangleleft$} (c0);

			\node[circle,fill, inner sep=2pt,label={[xshift=1mm]left:{\small$d_B$}}] at (-2, -2) (c4) {};
			\node[circle,fill, inner sep=2pt,label={[yshift=1mm]below:{\small$d_A$}}] at (0, -3) (c5) {};
			
			\draw[->, thick] (n2) to node[sloped,rotate=180] {\color{orange}$\blacktriangleleft$} (c4);
			\draw[->, thick] (n2) to (c5);

			\node[rectangle, inner sep=3pt,draw, thick,rounded corners=2pt] at (0, -1) (nl3) {\small $\langle l_A, l_B, q_2\rangle$};
			\node[rectangle, inner sep=3pt,draw, thick,rounded corners=2pt] at (0,-2) (nl4) {\small $\langle l_B, l_B, q_2\rangle$};

			\draw[->, thick] (c4) to[out=90,in=190] (nl3);
			\draw[->, thick] (c4) to[out=315,in=225]  (nl4);

			\draw[->, thick] (nl3) to node[sloped] {\color{orange}$\blacktriangleleft$} (c0);
			
			\draw[->, thick] (nl4) to node[sloped] {\color{orange}$\blacktriangleleft$} (c4);
			\draw[->, thick] (nl4) to (c5);

			\node[rectangle, inner sep=3pt,draw, thick,rounded corners=2pt] at (2, -2) (nl5) {\small $\langle l_A, l_A, q_2\rangle$};
			\node[rectangle, inner sep=3pt,draw, thick,rounded corners=2pt] at (2,-3) (nl6) {\small $\langle l_B, l_A, q_2\rangle$};

			\draw[->, thick] (c5) to (nl5);
			\draw[->, thick] (c5) to (nl6);

			\node[circle,fill, inner sep=2pt,label={[yshift=1mm]below:{\small$d_B$}}] at (1.5, 0) (c1) {};
			
			\draw[->, thick] (n0) to node[sloped] {\color{orange}$\blacktriangleleft$} (c0);
			\draw[->, thick] (n0) to (c1);

			\node[rectangle, inner sep=3pt,draw, thick,rounded corners=2pt] at (3.5,0) (n3) {\small $\langle l_A, l_B, q_1\rangle$};
			\node[rectangle, inner sep=3pt,draw, thick,rounded corners=2pt] at (6,0) (n4) {\small $\langle l_B, l_B, q_1\rangle$};

			\draw[->, thick] (c1) to (n3);
			\draw[->, thick,bend left] (c1) to (n4);

			\node[circle,fill, inner sep=2pt,label={[yshift=1mm]below:{\small$d_A$}}] at (4.75, -2) (c6) {};
			\node[circle,fill, inner sep=2pt,label={[yshift=1mm]below:{\small$d_B$}}] at (4.75, -3) (c7) {};
			
			\draw[->, thick] (n3) to (c6);
			\draw[->, thick] (n3) to (c7);
			\draw[->, thick] (n4) to (c6);
			\draw[->, thick] (n4) to (c7);
			
			\draw[->, thick] (nl3) to (c1);
			
			\draw[->, thick] (nl5) to (c6);
			\draw[->, thick] (nl5) to (c7);
			\draw[->, thick] (nl6) to (c6);
			\draw[->, thick] (nl6) to (c7);
			
			\draw[->, thick] (n1) to[out=0,in=130] (c1);

			\draw[red, very thick, dotted, rounded corners=5pt] (6,-0.8) rectangle (8, -4.2);

			\node[rectangle, inner sep=3pt,draw, thick,rounded corners=2pt] at (7,-1.2) (ne1) {\small $\langle l_A, l_A, q_3\rangle$};
			\node[rectangle, inner sep=3pt,draw, thick,rounded corners=2pt] at (7,-2) (ne2) {\small $\langle l_B, l_A, q_3\rangle$};
			\node[rectangle, inner sep=3pt,draw, thick,rounded corners=2pt] at (7,-3) (ne3) {\small $\langle l_A, l_B, q_3\rangle$};
			\node[rectangle, inner sep=3pt,draw, thick,rounded corners=2pt] at (7,-3.8) (ne4) {\small $\langle l_B, l_B, q_3\rangle$};

			\draw[->, thick] (c6) to (ne1);
			\draw[->, thick] (c6) to (ne2);
			\draw[->, thick] (c7) to (ne3);
			\draw[->, thick] (c7) to (ne4);
		\end{tikzpicture}
	}
	\vspace{-2mm}
		\subcaption{}\label{fig:instance}
	\end{subfigure}

	\caption{In \Cref{fig:ts}, we depict a transition system $\calT$ over $\ap = \{a\}$ using directions $\directions = \{d_A, d_B\}$. 
		 In \Cref{fig:dsa}, we give a DSA for the LTL body $\ltlG (a_{\pi_1} \leftrightarrow \ltlN a_{\pi_2})$ from \Cref{ex:running}, where we mark state $q_3$ as losing.
		In \Cref{fig:instance}, we sketch the QDec-POMDP $\calG_{\calT, \varphi}^\mathit{safe}$ constructed for the verification instance in \Cref{ex:running} (see \Cref{ex:planning} for details).
		}
\end{figure*}

We assume that $\ap$ is a fixed set of \emph{atomic propositions}. 
As the basic system model, we use finite-state transition systems (TS), which are tuples $\calT = (L, l_\mathit{init}, \directions, \kappa, \ell)$ where $L$ is a finite set of locations (we use ``locations'' to distinguish them from the ``states'' in a planning domain), $l_\mathit{init} \in L$ is an initial location, $\directions$ is a finite set of \emph{directions}, $\kappa : L \times \directions \to L$ is the transition function, and $\ell : L \to 2^\ap$ labels each location with an evaluation of the APs. 
Note that we use explicit directions in order to uniquely identify successor locations; we can easily model any transition function $L \to (2^L \setminus \{\emptyset\})$ using sufficiently many directions.
A path in $\calT$ is an infinite sequence $p \in L^\omega$ such that \textbf{(1)} $p(0) = l_\mathit{init}$, and \textbf{(2)} for every $j \in \nat$, there exists some $d \in \directions$ s.t.~$p(j+1) = \kappa(p(j), d)$.
We define $\paths(\calT) \subseteq L^\omega$ as the set of all paths in $\calT$.

\subsection{HyperLTL}

As the basic specification language for hyperproperties, we use HyperLTL, an extension of LTL with explicit quantification over (execution) paths \cite{ClarksonFKMRS14}.
Let $\pathVars = \{\pi, \pi', \ldots\}$ be a set of \emph{path variables}.
HyperLTL formulas are generated by the following grammar
\begin{align*}
	\psi &:= a_\pi \mid \psi \land \psi \mid \neg \psi \mid \ltlN \psi  \mid \psi \ltlU \psi \\
	\varphi &:=\quant \pi \ldot \varphi \mid \psi
\end{align*}
where $a \in \ap$, $\pi \in \pathVars$, and $\quant \in \{\forall, \exists\}$ is a quantifier.
We use the usual derived boolean constants and connectives $\mathit{true}, \mathit{false}, \lor, \to, \leftrightarrow$, and the temporal operators  \emph{eventually} ($\ltlF \psi := \mathit{true} \ltlU \psi$), and \emph{globally} ($\ltlG \psi := \neg \ltlF \neg \psi$).

Given a TS $\calT = (L, l_\mathit{init}, \directions, \kappa, \ell)$, we evaluate a HyperLTL formula in the context of a path assignment $\Pi : \pathVars \rightharpoonup L^\omega$ (mapping path variables to paths) and $j \in \nat$ as follows:
\begin{align*}
	\Pi, j &\models_\calT a_\pi &\text{iff } \quad&a \in \ell\big( \Pi(\pi)(j) \big)\\
	\Pi, j&\models_\calT \psi_1 \land \psi_2 \!\!\!\!\!\!\!\!&\text{iff } \quad &\Pi,j \models_\calT \psi_1 \text{ and } \Pi, j \models_\calT \psi_2\\
	\Pi, j &\models_\calT \neg \psi &\text{iff } \quad &\Pi, j \not\models_\calT \psi\\
	\Pi, j &\models_\calT \ltlN \psi &\text{iff } \quad &\Pi, j+1 \models_\calT \psi\\
	\Pi, j &\models_\calT \psi_1 \ltlU \psi_2 \!\!\! &\text{iff } \quad &\exists k \geq j\ldot \Pi, k \models_\calT \psi_2 \text{ and } \\
	&  \quad\quad\quad\quad\quad\quad\quad\quad\quad\quad \forall j \leq l < k\ldot \Pi, l \models_\calT \psi_1 \span \span\\
	\Pi, j &\models_\calT \quant \pi \ldot \varphi \!\!\!\!\!\!\!\! &\text{iff } \quad&\quant p \in \paths(\calT)  \ldot \Pi[\pi \mapsto p], j \! \models_\calT \! \varphi
\end{align*}
The atomic formula $a_\pi$ holds whenever $a$ holds in the current position $j$ on the path bound to $\pi$ (as given by $\ell$).
Boolean and temporal operators are evaluated as expected by updating the current evaluation position $j$, and quantification adds paths to $\Pi$.
We refer to \citet{Finkbeiner23} for details.
We say $\calT$ satisfies $\varphi$, written $\calT \models \varphi$, if $\emptyset, 0 \models_\calT \varphi$, where $\emptyset$ denotes the path assignment with empty domain.

\begin{example}\label{ex:running}
	As an example, consider the transition system $\calT$ in \Cref{fig:ts} and the HyperLTL formula $\varphi := \forall \pi_1\ldot \exists \pi_2\ldot \ltlG (a_{\pi_1} \leftrightarrow \ltlN a_{\pi_2})$.
	The formula expresses that for any path $\pi_1$, there exists some path $\pi_2$ that mirrors $\pi_1$ with a one-step delay. 
	It is easy to see that $\calT \models \varphi$.
\end{example}

\section{Verification via Planning}\label{sec:reduction}

We want to automatically verify that $\calT \models \varphi$.
To this end, we present a novel encoding into a planning problem, thus leveraging the extensive research and tool development within the planning community. 
As already outlined in \Cref{sec:overview}, our main idea is to interpret existential quantification in $\varphi$ as being resolved by an agent that picks transitions in $\calT$ to construct a path. 
Throughout this section, we assume that $\calT = (L, l_\mathit{init}, \directions, \kappa, \ell)$ is the fixed TS and $\varphi = \quant_1 \pi_1 \ldots \quant_n \pi_n\ldot \psi$ is the fixed HyperLTL formula over path variables $\pi_1, \ldots, \pi_n$.
We distinguish between temporal reachability (\Cref{sec:sub:reach}) and temporal safety (\Cref{sec:sub:safety}).

\subsection{Encoding for Temporal Reachability}\label{sec:sub:reach}

We first consider the case in which the LTL body of $\varphi$ expresses a \emph{temporal reachability property} (i.e., ``something good happens eventually'').

\paragraph{DFA.}

A deterministic finite automaton (DFA) over some alphabet $\Sigma$ is a tuple $\calA = (Q, q_0, \varrho, F)$ where $Q$ is a finite set of states, $q_0 \in Q$ is an initial state, $\varrho : Q \times \Sigma \to Q$ is a deterministic transition function, and $F \subseteq Q$ is a set of marked states.
An infinite word $u \in \Sigma^\omega$ is accepted by $\calA$ if the unique run \emph{eventually} reaches some state in $F$.
We say $\varphi$ is a \emph{temporal reachability formula} if $\psi$ (the LTL body of $\varphi$) is recognized by a DFA, i.e., some DFA $\calA_\psi = (Q_\psi, q_{0, \psi}, \varrho_\psi, F_\psi)$ over alphabet $2^{\ap \times \{\pi_1, \ldots, \pi_n\}}$  accepts exactly those infinite words that satisfy $\psi$ (recall that the atoms in $\psi$ have the form $a_{\pi_i} \in \ap \times \{\pi_1, \ldots, \pi_n\}$).

\paragraph{Planning Encoding.}
We write $\pathVars_\exists := \{\pi_i \mid \quant_i = \exists\}$ for the set of existentially quantified path variables in $\varphi$, and, analogously, $\pathVars_\forall$ for the set of universally quantified ones.

\begin{definition}\label{def:main}
	Define the QDec-POMDP $\calG_{\calT, \varphi}^\mathit{reach}$ as
	\begin{align*}
		\calG_{\calT, \varphi}^\mathit{reach} := \big(I, S, S_0,\allowbreak \{A_i\}, \delta, \{\Omega_i\}, \allowbreak\{\omega_i\},\allowbreak G\big),
	\end{align*}
	where
	\begin{align*}
		I &:= \big\{i \mid \pi_i\in \pathVars_\exists\big\},\\
		S &:= \big\{\langle l_1, \ldots, l_n, q \rangle \mid q \in Q_\psi \land l_1, \ldots, l_n \in L\big\},\\
		S_0 &:= \big\{\langle l_\mathit{init}, \ldots, l_\mathit{init}, q_{0, \psi} \rangle\big\},\\
		A_i &:= \directions,\\
		\Omega_i &:= \big\{\langle l_1, \ldots, l_i \rangle \mid l_1, \ldots, l_i \in L\big\},\\
		G &:= \big\{\langle l_1, \ldots, l_n, q \rangle \mid q \in F_\psi \land l_1, \ldots, l_n \in L\big\},
	\end{align*}
	the transition function $\delta$ is defined as
	\begin{align*}
						\delta\big(\langle &l_1, \ldots, l_n, q \rangle, \{d_i \in \directions\}_{\pi_i \in \pathVars_\exists}\big) := \\
		&\Big\{ \langle \kappa(l_1, d_1), \ldots, \kappa(l_n, d_n), q' \rangle \mid \{d_i \in \directions\}_{\pi_i \in \pathVars_\forall} \, \land \\
		&\textstyle\quad q' = \varrho_\psi\big(q, \bigcup_{i=1}^n \big\{ (a, \pi_i) \mid a \in \ell(l_i)  \big\} \big) \Big\},
	\end{align*} 
	and the observation functions $\{\omega_i\}$ are defined as
	\begin{align*}
		\omega_i(\langle l_1, \ldots, l_n, q \rangle) := \langle l_1, \ldots, l_i \rangle.
	\end{align*}
\end{definition}

Let us step through this definition.
As already hinted in \Cref{sec:overview}, we add one agent $i$ for each existentially qualified path variable $\pi_i \in \pathVars_\exists$.
Each state has the form $\langle l_1, \ldots, l_n, q \rangle$ and tracks a current location for each of the paths (where $l_i \in L$ is the current location for path $\pi_i$), and the current state of $\calA_\psi$. 
Intuitively, the planning problem simulates $\pi_1, \ldots, \pi_n$ by keeping track of their current location ($l_1, \ldots, l_n$), and letting the actions chosen by the agents (for existentially quantified paths) or the non-determinism (for universally quantified paths) fix the next location.
We start each $\pi_i$ in the initial location $l_\mathit{init}$ and start the run of $\calA_\psi$ in the initial state $q_{0, \psi}$. 
The actions of each agent then directly correspond to directions in $\calT$. 
When given a joint action $\{d_i\}_{\pi_i\in \pathVars_\exists}$ (i.e., a direction for each existentially quantified path), the transition function considers \emph{all possible} directions for universally quantified paths ($\{d_i\}_{\pi_i\in \pathVars_\forall}$) and updates each location $l_i$ based on the direction $d_i$.
Existentially quantified paths thus follow the direction selected by the respective agent, and universally quantified ones follow a non-deterministically chosen direction. 
In each step, we also update the state of $\calA_\psi$:
For each $1 \leq i \leq n$, we collect all APs that hold in the current location ($\ell(l_i)$) and index them with $\pi_i$, thus obtaining a letter in $2^{\ap \times \{\pi_1, \ldots, \pi_n\}}$ which we feed to the transition function of $\calA_\psi$. 
As argued in \Cref{sec:overview}, each agent $i$ controlling $\pi_i \in \pathVars_\exists$ may only observe the paths $\pi_1, \ldots, \pi_i$ quantified \emph{before} $\pi_i$, so the observation set $\Omega_i$ of agent $i$ consist exactly of tuples of the form $\langle l_1, \ldots, l_i \rangle$ and the observation function $\omega_i$ projects each state to the observable locations. 
Lastly, the goal consists of all states in which the automaton has reached one of $\calA_\psi$'s marked states.

\begin{restatable}{theorem}{th}\label{theo:th1}
Assume $\varphi$ is a temporal reachability formula.
If $\calG_{\calT, \varphi}^\mathit{reach}$ admits a strong plan, then $\calT \models \varphi$. 
\end{restatable}

\begin{proof}[Proof Sketch]
	We can use the policies in a strong plan for $\calG_{\calT, \varphi}^\mathit{reach}$ to construct \emph{Skolem functions} for existentially quantified paths in $\varphi$.
	The full proof is provided in \ifFull{\Cref{app:proof}}{the full version \cite{fullVersion}}.
\end{proof}

\begin{remark}
	In the HyperLTL semantics, an existentially quantified path $\pi_i$ is chosen based on \emph{all} paths quantified before $\pi_i$. 
	In our encoding, we abstract this selection with a step-wise action selection, so the agents can construct existentially quantified paths based only on the \emph{prefixes} of previously quantified paths.
	This lack of information leads to incompleteness, i.e., in some cases $\calT \models \varphi$, but $\calG_{\calT, \varphi}^\mathit{reach}$ does \emph{not} admit a strong plan.
	We can counteract this incompleteness by providing the agents with information about the future behavior of universally quantified executions using prophecies \cite{BeutnerF22}. 
\end{remark}

\subsection{Encoding for Temporal Safety}\label{sec:sub:safety}

We can also handle the case in which the LTL body of $\varphi$ denotes a \emph{temporal safety property}.

\paragraph{DSA.}

A deterministic safety automaton (DSA) over alphabet $\Sigma$ is a DFA-like tuple $\calA = (Q, q_0, \varrho, F)$.
An infinite word $u \in \Sigma^\omega$ is accepted by $\calA$ if the unique run \emph{never} visits a state in $F$.
We say $\varphi$ is a \emph{temporal safety formula} if $\psi$ is recognized by a DSA $\calA_\psi = (Q_\psi, q_{0, \psi}, \varrho_\psi, F_\psi)$.
For example, the \ref{eq:gni} formula in \Cref{sec:intro} and the HyperLTL formula from \Cref{ex:running} are temporal safety formulas.

\paragraph{Planning Encoding.}

Different from reachability properties, safety properties reason about infinite executions (and not only finite prefixes thereof).
To encode this as a planning problem, we add special sink states $s_\mathit{win}$ and $s_\mathit{lose}$, and mark $s_\mathit{win}$ as the unique goal state. 
From any state $\langle l_1, \ldots, l_n, q \rangle$ where $q \in F_\psi$, we then deterministically move to $s_\mathit{lose}$.
From any state $\langle l_1, \ldots, l_n, q \rangle$ where $q \not\in F_\psi$, we extend the transitions from \Cref{def:main} with an additional non-deterministic transition to $s_\mathit{win}$.
The agents can thus never ensure a visit to $s_\mathit{win}$, but a strong \emph{cyclic} plan guarantees that we never visit a losing state in $F_\psi$.
We denote the resulting QDec-POMPD with $\calG_{\calT, \varphi}^\mathit{safe}$; a full description can be found in \ifFull{\Cref{app:safety}}{the full version \cite{fullVersion}}.

\begin{restatable}{theorem}{thh}\label{theo:th2}
	Assume $\varphi$ is a temporal safety formula.
	If $\calG_{\calT, \varphi}^\mathit{safe}$ admits a strong cyclic plan, then $\calT \models \varphi$. 
\end{restatable}

\begin{example}\label{ex:planning}
	We consider the verification instance $(\calT, \varphi)$ from \Cref{ex:running}.
	In \Cref{fig:dsa}, we depict a DSA for the LTL body of $\varphi$.
	In \Cref{fig:instance}, we sketch the resulting QDec-POMDP $\calG_{\calT, \varphi}^\mathit{safe}$. 
	Each action from $\directions = \{d_A, d_B\}$ for agent $2$ (controlling the existentially quantified path $\pi_2$) non-deterministically leads to two successor states, which we visualize using immediate decision nodes labeled by directions.
	For example, in the initial state $\langle l_A, l_A, q_0\rangle$, the action (aka.~direction) $d_A$ non-deterministically leads to states $\langle l_A, l_A, q_1\rangle$ and $\langle l_B, l_A, q_1\rangle$. 
	For simplicity, we omit $s_\mathit{win}$ and  $s_\mathit{lose}$: All states where the automaton has reached the losing state $q_3$ (surrounded by the red dashed box) transition to $s_\mathit{lose}$, and all other states have an additional non-deterministic transition to $s_\mathit{win}$.
	The QDec-POMDP $\calG_{\calT, \varphi}^\mathit{safe}$ admits a strong cyclic plan by always following the action marked by the orange arrow, proving that $\calT \models \varphi$. 
\end{example}

\subsection{Encoding for Full HyperLTL}

Our construction can also be extended to handle \emph{full} HyperLTL by reducing to planning problems with temporal goals specified in LTL \cite{CamachoTMBM17,CamachoM19}. 
In this paper, we restrict our construction  to the case of reachability and safety properties as \textbf{(1)} this suffices for almost all properties of interest, and \textbf{(2)} it allows us to employ the multiplicity of automated planners that yield strong (cyclic) plans for non-temporal objectives.

\subsection{Factored Representation}

In our construction, we used an explicit-state (flat) representation of the problem with $|L|^n \cdot |Q_\psi|$ many states. 
In practice, many planning formats (e.g., STRIPS, PDDL, SAS) allow for a \emph{factored} description of the state space, using roughly $n \cdot |L| + |Q_\psi|$ many fluents that track the current location of each path \emph{individually}.
For example, the QDec-POMDP from \Cref{fig:instance} can be represented compactly by tracking the locations of $\pi_1$ and  $\pi_2$ individually. 
The possibility of using a factored representation is a core motivation for using planning tools for HyperLTL verification.

\subsection{Classical, FOND, and POND Planning}\label{sec:sub:easier}

In general, our encoding yields a planning problem that combines multiple agents, non-determinism, and partial observations. 
In many situations, however, the resulting problem does not require all these features:
\textbf{(1)} For $\exists^*$ properties, the planning problem is classical, i.e., it consists of a single agent (controlling all paths), deterministic actions, and full information. 
\textbf{(2)} For $\forall^*\exists^*$ properties (e.g., \ref{eq:gni}), the problem involves a single agent (controlling all existentially quantified paths) acting under full information (FOND planning).
\textbf{(3)} For $\forall^*\exists^*\forall^*$ properties, the problem involves a single agent acting under partial observations (POND planning).

\section{Implementation and Experiments}\label{sec:implementation}

\begin{table}[!t]
	\centering
	
		\small
		\renewcommand{\arraystretch}{1.2}
		\begin{tabular}{@{\hspace{0mm}}l@{\hspace{2mm}}l@{\hspace{2mm}}l@{\hspace{2mm}}c@{\hspace{1mm}}c@{\hspace{2mm}}c@{\hspace{1mm}}c@{\hspace{0mm}}}
			\toprule 
			& & & \multicolumn{2}{c}{$\boldsymbol{\exists}\boldsymbol{\exists}$} & \multicolumn{2}{c}{$\boldsymbol{\forall}\boldsymbol{\exists}$} \\
			\cmidrule[0.6pt](l{0mm}r{2mm}){4-5}
			\cmidrule[0.6pt](l{-1mm}r{0mm}){6-7}
			\textbf{Model} & \textbf{Size PG} & \textbf{Size PDDL} & $\boldsymbol{t}_\texttt{PG}$ & $\boldsymbol{t}_\tool{}$ & $\boldsymbol{t}_\texttt{PG}$ & $\boldsymbol{t}_\tool{}$  \\
			\midrule[0.6pt]
			\textsc{bakery3} & 31016.4 & 8.0/75.7 & 4.8 & \textbf{1.2} & 4.5 & \textbf{0.9} \\
			\textsc{bakery5} & 614.6 & 7.7/19.0 & 0.6 & \textbf{0.5}& \textbf{0.7} & 1.1 \\
			\textsc{mutation} & 1807.5 & 8.5/9.8 & 0.7 & \textbf{0.3} & 0.6 & \textbf{0.4} \\
			\textsc{ni\_c} & 37.5 & 7.7/13.5 & \textbf{0.2} & 0.3 &  0.5 & \textbf{0.4} \\
			\textsc{ni\_i} & 948.3 & 8.5/104.3 & 0.6 & \textbf{0.4} & \textbf{0.7} & 1.3 \\
			\textsc{nrp\_c} & 688.3 & 7.7/23.0 & \textbf{0.5} & \textbf{0.5} & 0.5 & \textbf{0.4} \\
			\textsc{nrp\_i} & 1018.6 & 7.7/22.2 & 0.5 & \textbf{0.3} & \textbf{0.4} & \textbf{0.4} \\
			\textsc{snark\_con} & 105854.7 & 7.7/192.1 & 19.2 &\textbf{6.9} & 21.0 & \textbf{5.1} \\
			\textsc{snark\_seq} & 17415.6 & 8.0/32.4 & 3.3 & \textbf{0.6} & 5.6 & \textbf{2.4} \\
			\bottomrule
		\end{tabular}
		
	\caption{We compare \tool{} with a PG-based encoding. We list the size of the PG, the number of actions/objects in \tool{}'s PDDL encoding, and the verification times in seconds (averaged over multiple runs).
	}\label{tab:res1}

\end{table}

We have implemented our encoding for $\forall^*\exists^*$ HyperLTL formulas in a prototype called \tool{}.
Our tool produces FOND planning instance in an extension of PDDL, featuring  \texttt{(oneof p1 ... pn)} expressions in action effects; A format widely supported by many FOND planners. 
We compare \tool{} against the parity-game (PG) based encoding for $\forall^*\exists^*$ properties \cite{BeutnerF22}.
For our experiments, we collect the 10 NuSMV models from \citet{HsuSB21} and generate random formulas of the form $\exists \pi\ldot \exists \pi' \ldot \ltlF \psi$ and $\forall \pi\ldot \exists \pi' \ldot \ltlG \psi$ where $\psi$ is a temporal-operator-free formula. 
As remarked in \Cref{sec:sub:easier}, for the $\exists\exists$ properties, \tool{} produces classical planning problems, which we solve using \texttt{Scorpion} \cite{SeippH18}.
The $\forall\exists$ properties yield FOND planning problems, which we solve using \texttt{MyND} \cite{MattmullerOHB10}. 
In \Cref{tab:res1}, we report the average size of the encodings, as well as the time taken for the $\exists\exists$ and $\forall\exists$ properties.
As noted in \Cref{sec:relatedWork}, the size of the PG is exponential in the number of paths, whereas the PDDL description is \emph{factored} and delegates the exploration to the planner.
Consequently, we observe that off-the-shelf planners perform competitively compared to explicit-state PG solvers.

\section{Conclusion}

We have presented a novel application of non-deterministic planning: the verification of hyperproperties. 
Our encoding is applicable to HyperLTL formulas with arbitrary quantifier prefixes, and our preliminary experiments show that off-the-shelf planners constitute a competitive alternative to explicit-state games. 
Moreover, further development of non-deterministic planners (for which our work provides additional incentives) directly improves our verification pipeline.  

\section*{Acknowledgments}

This work was supported by the European Research Council (ERC) Grant HYPER (101055412), and by the German Research Foundation (DFG) as part of TRR 248 (389792660).

\bibliography{references}

\iffullversion

\newpage

\appendix

\section{Correctness Proof}\label{app:proof}

\th*
\begin{proof}
	Let $\varphi = \quant_1 \pi_1\ldots \quant_n \pi_n\ldot \psi$ be the fixed HyperLTL formula and assume that $\calG_{\calT, \varphi}^\mathit{reach}$ admits a strong plan.
	
	\paragraph{Skolem Functions for HyperLTL.}
	To show that $\calT \models \varphi$, we construct a \emph{Skolem function} $\xi_i : \paths(\calT)^{i-1} \to \paths(\calT)$ for each $i$ where $\quant_i = \exists$.
	Intuitively,  $\xi_i$ provides a concrete witness path for $\pi_i$ when given the $i-1$ paths assigned to path variables $\pi_1, \ldots, \pi_{i-1}$ \cite{shoenfield2018mathematical}.
	Given a family of Skolem functions $\{\xi_i\}_{\pi_i \in \pathVars_\exists}$ (one for each $\pi_i \in \pathVars_\exists$), we say a path combination $(p_1, \ldots, p_n) \in \paths(\calT)^n$ \emph{is permitted by $\{\xi_i\}_{\pi_i \in \pathVars_\exists}$}, iff $p_i = \xi_i(p_1, \ldots, p_{i-1})$ whenever $\pi_i \in \pathVars_\exists$.
	That is, for all existentially quantified paths, we use the Skolem function but choose arbitrary paths for all universally quantified ones. 
	We say $\{\xi_i\}_{\pi_i \in \pathVars_\exists}$ is \emph{winning} if, for all path combinations $(p_1, \ldots, p_n)$ that are  permitted by $\{\xi_i\}_{\pi_i \in \pathVars_\exists}$, we have $[\pi_1 \mapsto p_1, \ldots, \pi_n \mapsto p_n], 0 \models_\calT \psi$, i.e., the body of $\varphi$ is satisfied when binding $p_i$ to $\pi_i$ for each $1 \leq i \leq n$. 
	Using standard first-order techniques it is easy to see that if some winning family of Skolem functions exists, then $\calT \models \varphi$ \cite{shoenfield2018mathematical}.
	
	\paragraph{Skolem Function Construction.}
	We construct the Skolem functions $\{\xi_i\}_{\pi_i \in \pathVars_\exists}$ using the fact that $\calG_{\calT, \varphi}^\mathit{reach}$ admits a strong plan.
	Let $\{f_i\}_{\pi_i \in \pathVars_\exists}$ be a joint policy that is a strong plan in $\calG_{\calT, \varphi}^\mathit{reach}$. 
	Here, each $f_i$ is a function $f_i : \Omega_i^+ \to \directions$ where $\Omega_i = \{\langle l_1, \ldots, l_i\rangle \mid l_1, \ldots, l_i \in L\}$ (cf.~\Cref{def:main}). 
	Recall that being a strong plan means that there exists a $N \in \nat$, such that for any $p \in \interactions(\{f_i\})$ with $|p| \geq N$, $p$ visits some state in $G$.
	 
	The intuition is to construct $\xi_i$ from $f_i$ by querying it one prefixes. 
	The Skolem function $\xi_i$ will see $i-1$ complete paths in $\calT$ and outputs an entire path.
	In contrast, $f_i$ observes the locations of the first $i-1$ paths, i.e., has only seen a prefix of the final paths, and needs to fix a direction. 
	We will use $f_i$ to construct $\xi_i$ by iteratively querying it on those prefixes. 
	Given $p_1, \ldots, p_{i-i} \in \paths(\calT)$, we define the path $\xi_i(p_1, \ldots, p_{i-i})  \in \paths(\calT)$ pointwise as follows:
	Initially, we set 
	\begin{align*}
		\xi_i(p_1, \ldots, p_{i-i})(0) = l_\mathit{init},
	\end{align*}
	i.e., the initial location of $\calT$. 
	For each $k \in \nat$, we then define $\xi_i(p_1, \ldots, p_{i-i})(k+1)$ as 
	\begin{align*}
			\kappa\bigg(\xi_i(p_1, &\ldots, p_{i-i})(k), \\
			f_i \Big(&\big\langle p_1(0), \ldots, p_{i-1}(0), \xi_i(p_1, \ldots, p_{i-i})(0) \big\rangle\\
			&\big\langle p_1(1), \ldots, p_{i-1}(1), \xi_i(p_1, \ldots, p_{i-i})(1) \big\rangle\\
			&\quad\quad\cdots \\
			&\big\langle p_1(k), \ldots, p_{i-1}(k), \xi_i(p_1, \ldots, p_{i-i})(k) \big\rangle\Big) \bigg).
	\end{align*}
	Let us unpack this definition: 
	To define the $(k+1)$th step on $\xi_i(p_1, \ldots, p_{i-i})$, we need to transition from the $k$th step ($\xi_i(p_1, \ldots, p_{i-i})(k)$) along some direction. 
	This direction will be chosen by the local policy $f_i : \Omega_i^+ \to \directions$, so we need to construct the prefix in $ \Omega_i^+$, where each element in $\Omega_i$ has the form $\langle l_1, \ldots, l_i\rangle$ (cf.~\Cref{def:main}).
	The idea is that this prefix consists exactly of the states $\langle l_1, \ldots, l_i\rangle$ that $\xi_i(p_1, \ldots, p_{i-i})$ -- together with the fixed paths $p_1, \ldots, p_{i-1}$ -- has traversed so far.
	That is, even though we already know the entire paths $p_1, \ldots, p_{i-1}$, we acts as if we only knew the prefix up to step $k$. 
	
	By following this construction, we obtain a path $\xi_i(p_1, \ldots, p_{i-i}) \in \paths(\calT)$, and have thus defined $\xi_i$. 
	
	\paragraph{Skolem Functions Are Winning.}
	
	It remains to argue that $\{\xi_i\}_{\pi_i \in \pathVars_\exists}$ is winning. 
	For this, let $(p_1, \ldots, p_n)$ be a path combination that is permitted by $\{\xi_i\}_{\pi_i \in \pathVars_\exists}$. 
	We need to argue that $[\pi_1 \mapsto p_1, \ldots, \pi_n \mapsto p_n], 0 \models_\calT \psi$. 
	Let $q_0q_1 \cdots \in Q_\psi^\omega$ be the unique run of $\calA_\psi$ on paths $(p_1, \ldots, p_n)$, which we can define inductively by
	\begin{align*}
		q_0 &:= q_{0, \psi}\\
		q_{k+1} &:= \varrho_\psi\Big(q_{k}, \bigcup_{i=1}^n \Big\{ (a, \pi_i) \mid a \in \ell\big(p_i(k)\big)  \Big\} \Big).
	\end{align*}
	By construction of $\calA_\psi$, $[\pi_1 \mapsto p_1, \ldots, \pi_n \mapsto p_n], 0 \models_\calT \psi$ is now equivalent to the fact that the unique run $q_0q_1\cdots$ is accepted by $\calA_\psi$, i.e., eventually visits a state in $F_\psi$. 
	To show that this is indeed the case, we show that
	\begin{align}\label{eq:path}
		\langle p_1(0), \ldots, p_n(0), q_0 \rangle \langle p_1(1), \ldots, p_n(1), q_1 \rangle \cdots \in S^\omega
	\end{align}
	is a (for simplicity, infinite) path in $\calG_{\calT, \varphi}^\mathit{reach}$ that is allowed by the joint policy $\{f_i\}_{\pi_i \in \pathVars_\exists}$, i.e., every prefix is contained in $\interactions(\{f_i\}_{\pi_i \in \pathVars_\exists})$.
	By assumption we have that combination  $(p_1, \ldots, p_n)$ is permitted by our Skolem functions $\{\xi_i\}_{\pi_i \in \pathVars_\exists}$. 
	For all existentially quantified path variables $\pi_i \in \pathVars_\exists$, $p_i$ (which satisfies $pi_ i = \xi_i(p_1, \ldots, p_{i-1})$) now follows exactly the directions chose by $f_i$. 
	Likewise, $\calG_{\calT, \varphi}^\mathit{reach}$ does not restrict the directions for universally quantified paths $\pi_i \in \pathVars_\forall$, as their direction is chosen non-deterministically (cf.~\Cref{def:main}).
	Lastly, the automaton states $q_0q_1\cdots$ are defined exactly as in the transitions of $\calG_{\calT, \varphi}^\mathit{reach}$.
	As $(p_1, \ldots, p_n)$ is permitted by $\{\xi_i\}_{\pi_i \in \pathVars_\exists}$, any prefix of \Cref{eq:path} is compatible with the joint policy $\{f_i\}_{\pi_i \in \pathVars_\exists}$.
	As $\{f_i\}_{\pi_i \in \pathVars_\exists}$ is a strong plan in $\calG_{\calT, \varphi}^\mathit{reach}$, at some timepoint $M \in \nat$, we get that $\langle p_1(M), \ldots, p_n(M), q_M \rangle \in G$ which, by definition of $G$ (cf.~\Cref{def:main}), is equivalent to $q_M \in F_\psi$.
	This already implies that $q_0q_1 \cdots$ is accepting (i.e., visits an accepting state), which, in turn, implies that $[\pi_1 \mapsto p_1, \ldots, \pi_n \mapsto p_n], 0 \models_\calT \psi$.
	As this holds for any combination $(p_1, \ldots, \pi_n)$ permitted by $\{\xi_i\}_{\pi_i \in \pathVars_\exists}$, we get that $\{\xi_i\}_{\pi_i \in \pathVars_\exists}$ is a winning family of Skolem functions, so $\calT \models \varphi$ as required. 
\end{proof}

\section{Encoding for Temporal Safety}\label{app:safety}

In this section, we provide a formal construction for the verification of temporal safety formulas.
We again assume that $\calT = (L, l_\mathit{init}, \directions, \kappa, \ell)$ is a fixed TS and $\varphi = \quant_1 \pi_1 \ldots \quant_n \pi_n\ldot \psi$ is a fixed safety HyperLTL formula.
We let $\calA_\psi = (Q_\psi, q_{0, \psi}, \varrho_\psi, F_\psi)$ over $2^{\ap \times \{\pi_1, \ldots, \pi_n\}}$ be a DSA that accepts exactly the words that safety $\psi$ in a safety semantics. 

\begin{definition}\label{def:safe}
	Define the QDec-POMDP $\calG_{\calT, \varphi}^\mathit{safe}$ as
	\begin{align*}
		\calG_{\calT, \varphi}^\mathit{safe} := \big(I, S, S_0,\allowbreak \{A_i\}, \delta, \{\Omega_i\}, \allowbreak\{\omega_i\},\allowbreak G\big),
	\end{align*}
	where
	\begin{align*}
		I &:= \big\{i \mid \pi_i\in \pathVars_\exists\big\}\\
		S &:= \big\{\langle l_1, \ldots, l_n, q \rangle \mid q \in Q_\psi \land l_1, \ldots, l_n\in L\big\} \;\cup \\
		&\quad\quad \quad  \big\{s_\mathit{win}, s_\mathit{lose}\big\},\\
		S_0 &:= \big\{\langle l_\mathit{init}, \ldots, l_\mathit{init}, q_{0, \psi} \rangle\big\},\\
		A_i &:= \directions,\\
		\Omega_i &:= \big\{\langle l_1, \ldots, l_i \rangle \mid l_1, \ldots, l_i \in L\big\},\\
		G &:= \big\{s_\mathit{win}\big\},
	\end{align*}
	the transition function $\delta$ is defined as
	\begin{align*}
		&\delta\big(s_\mathit{win}, \{d_i\}_{\pi_i\in \pathVars_\exists}\big) := s_\mathit{win}\\
		&\delta\big(s_\mathit{lose}, \{d_i\}_{\pi_i\in \pathVars_\exists}\big) := s_\mathit{lose}\\
		&\delta\big(\langle l_1, \ldots, l_n, q \rangle, \{d_i \in \directions\}_{\pi_i \in \pathVars_\exists}\big) := \\
		&\quad\quad\Big\{ \langle \kappa(l_1, d_1), \ldots, \kappa(l_n, d_n), q' \rangle \mid \{d_i \}_{\pi_i \in \pathVars_\forall} \, \land \\
		&\quad\quad\quad q' = \varrho_\psi\big(q, \bigcup_{i=1}^n \big\{ (a, \pi_i) \mid a \in \ell(l_i)  \big\} \big) \Big\} \cup  \{s_\mathit{win}\}  \\
		&\quad\quad\quad\quad\quad \text{when } q \not\in F_\psi\\
		&\delta\big(\langle l_1, \ldots, l_n, q \rangle, \{d_i \}_{\pi_i \in \pathVars_\exists}\big) := s_\mathit{lose}\\
		&\quad\quad\quad\quad\quad \text{when } q \in F_\psi
	\end{align*} 
	and the observation functions $\{\omega_i\}$ are defined as
	\begin{align*}
		\omega_i(\langle l_1, \ldots, l_n, q \rangle) := \langle l_1, \ldots, l_i \rangle.
	\end{align*}
\end{definition}

Compared to \Cref{def:main}, we add two new states $s_\mathit{win}$ and $s_\mathit{lose}$. 
The objective of the agents is to reach $s_\mathit{win}$ and avoid $s_\mathit{lose}$.
For both of these states, we add a self-loop, so once either of the two is reached, it is never left. 

In the transition from regular states $\langle l_1, \ldots, l_n, q\rangle$ we then distinguish if $q \in F_\psi$. 
If $q \in F_\psi$ so the automaton has reach a bad state and will reject by deterministically transitioning to $s_\mathit{lose}$ (note that from $s_\mathit{lose}$ the goal state $s_\mathit{win}$ is unreachable).
On the other hand, if $q \not\in F_\psi$, we progress similar as in \Cref{def:main}. 
The only difference is that we \emph{add} $s_\mathit{win}$ as an additional non-deterministic outcome.
As long as the agent manage to stay within states where $q \not\in F_\psi$, \emph{some} non-deterministic outcome can thus reach the goal state $s_\mathit{win}$. 
A strong \emph{cylic} plan always ensures the possibility to win and thus always stays within such good states.

\thh*
\begin{proof}
	The proof is similar to \Cref{theo:th1}. 
	We use the same construction for the Skolem functions of each existentially quantified path. 
	Following the intuition behind \Cref{def:safe}, the resulting paths are such that, within the state-space of  $\calG_{\calT, \varphi}^\mathit{safe}$ any path only visits states of the form $\langle l_1, \ldots, l_n, q \rangle$ with $q \not\in F_\psi$ or $s_\mathit{win}$; if some path would reach a state $\langle l_1, \ldots, l_n, q \rangle$ with $q \not\in F_\psi$ the joint policy would not be winning.
	In particular, the DSA $\calA_\psi$ that tracks the acceptance of $\psi$ only remains in states outside of $F_\psi$, showing $\calT \models \varphi$ as required. 
\end{proof}

\fi

\end{document}